\documentclass[12pt]{amsart}%
\usepackage{color}
\usepackage[right=2.7cm,left=2.7cm,top=3.8cm,bottom=3.5cm]{geometry}

\usepackage{amsfonts}
\usepackage{bbold}
\usepackage{dsfont}
\usepackage{amssymb}
\usepackage{flafter}
\usepackage{amsthm}
\usepackage{mathtools}

\usepackage{float}
\usepackage[pdftex]{graphicx}
\usepackage{amsmath}
\usepackage{enumitem}
\usepackage{amsfonts}
\usepackage{dsfont}
\usepackage{amssymb}
\usepackage{flafter}
\usepackage[pdffitwindow=false, pdfview=FitH, pdfstartview=FitH,
plainpages=false]{hyperref}%
\providecommand{\U}[1]{\protect \rule{.1in}{.1in}}

\newtheorem{examples}{Examples}[section]
\newtheorem{lemma}{Lemma}[section]
\newtheorem{proposition}{Proposition}[section]

\newtheorem{definition}{Definition}[section]
\newtheorem{theorem}{Theorem}[section]
\newtheorem{remark}{Remark}[section]
\newtheorem{example}{Example}[section]

\newcommand{\Z}{{\mathbb Z}}
\newcommand{\ZZ}{\Z_{p^r}}

\newcommand{\A}{\mathcal{A}_{p}}

\newcommand{\F}{{\mathbb F}}

\newcommand{\C}{{\mathcal C}}

\usepackage{amssymb}

\makeindex

\newcommand{\vu}{{{\bf u}}}
\newcommand{\vv}{{{\bf v}}}
\newcommand{\vw}{{{\bf w}}}

\usepackage{etoolbox}
\makeatletter
\@ifclassloaded{amsart}{
}{
	\newcommand{\@keywords}{}
	\providecommand{\keywords}[1]{\gdef\@keywords{#1}}
	\AtEndEnvironment{abstract}{%
		\ifx\@keywords\@empty\else
		\par\vspace{1ex}\noindent\textbf{Mots-clés : } \@keywords\par\vspace{1ex}%
		\fi
	}
}
\makeatother
\begin{document}

\title{On Catastrophicity  of Convolutional Codes and their encoders over $\ZZ$}

\author{Mohammed El Oued}
\address{Department of Mathematics, Higher Institute of Computer Science and Mathematics, University of Monastir, Tunisia}
\email{mohamed.eloued@isimm.rnu.tn}
\maketitle

\begin{abstract} This paper investigates the existence of minimal $p$-encoders for convolutional codes over $\mathbb{Z}_{p^r}$, where $p$ is a prime. This addresses a conjecture from \cite{k}, which posits that every such code admits a minimal $p$-encoder, implying that all convolutional codes over $\mathbb{Z}_{p^r}$ are noncatastrophic  when input sequences are restricted to coefficients in $\{0, \dots, p-1\}$.  Our contributions include the introduction of a new polynomial invariant that characterizes free codes, which enables us to establish a necessary and sufficient condition for a free code over $\mathbb{Z}_{p^r}$ to be noncatastrophic in the usual sense (where input coefficients are from $\mathbb{Z}_{p^r}$). Based on these findings, we affirm the conjecture by providing a constructive method for obtaining a minimal $p$-encoder for any convolutional code over $\mathbb{Z}_{p^r}$.

\end{abstract}

\section{Introduction}
 Convolutional codes over rings   were introduced by  Massey and Mittelholzer \cite{ma}, with a specific focus on the ring  $\Z_M$. They demonstrated that linear
codes over $\Z_M$ are particularly suitable for phase modulation. Johannesson, Wan and Wittenmark \cite{z} further the structural analysis of these
codes, which exhibit significantly different behavior compared to codes over
fields. They introduced the concepts of right-invertible, noncatastrophic, and  basic encoders and analyzed the algebraic structure of these classes of encoders. Fagnani and Zampieri \cite{f}  extended this analysis   presenting a complete theory
of convolutional codes over $\Z_{p^r}$ in the usual case where the input sequence is
a free module. More recently, Kuijper and Pinto introduced  \cite{k}   the notions of $p$-encoder and minimal $p$-encoder for convolutional codes over $\ZZ$, thereby
defining several invariants for these codes, such as $p$-dimension and $p$-degree,
generalizing the existing notions for codes over fields.\\
In the case of codes over fields, any linear convolutional code admits a noncatastrophic encoder, i.e, a generator matrix $G(D)$ with the property that the output $\vu(D) G(D)$ is of finite weight only if the input $\vu(D)$ is also so. Thus,
catastrophicity in this context is a property of the encoder rather than the code
itself.  However, for codes over finite rings, the existence of noncatastrophic
encoders has been the subject of several studies. Massey and Mittelholzer \cite{z}
showed that a polynomial matrix $G(D)\in \Z[D]^{k\times n}$ is noncatastrophic if and
only if its projection $ \overline{G(D)}$ in $\Z_p[D]^{k\times n}$ is also noncatastrophic. Fagnani and
Zampieri \cite{f} proved that any free convolutional code over $\Z_{p^r}$ admits a noncatastrophic rational matrix. Nevertheless, there exist codes for which every
polynomial encoder is catastrophic, indicating that catastrophicity can be an
intrinsic property of the code.\\
The occurrence of catastrophic codes motivated the introduction of the concept of minimal $p$-encoder for convolutional codes over $\Z_p$ in \cite{k}. A minimal $p$-encoder $G(D)$ for a code $\C$ is defined as a noncatastrophic $p$-encoder whose rows form a reduced $p$-basis for $\C$. The authors conjectured that any convolutional code over $\ZZ$ admits a minimal $p$-encoder. In this work, we address and resolve this conjecture by  presenting  a constructive proof  of the existence of a  minimal $p$-encoder for any convolutional code over $\ZZ$.\\
The remainder of this paper is organized as follows:  next section is devoted to the exposition of some preliminaries on convolutional codes over $\ZZ$, specifically recalling  the notions of   $p$-basis, $p$-encoder and minimal $p$-encoder. Section 3 deals with full row rank polynomial matrices over a finite field, focusing on basic and catastrophic matrices. Section 4  studies the catastrophicity of polynomial matrices and free codes over $\ZZ$. Finally,  Section 5 is devoted to the existence and constriction   of minimal $p$-encoders for any convolutional code over $\ZZ$.

\section{Algebraic preliminaries }
\label{sec:2}
This section establishes the foundational algebraic concepts and notation used throughout this paper, particularly concerning rings of series and rational functions, and the definition of convolutional codes over $\ZZ$.\\
Let $\ZZ((D))$ be the ring of Laurent series over $\ZZ$. We denote by $\ZZ[D]$ the ring of polynomials with coefficients in $\ZZ$ and by $\ZZ(D)$ the  ring of rational functions defined in $\ZZ$. More precisely, $\ZZ(D)$  is the set of rational functions $\frac{p(D)}{q(D)}$, where $p(D),q(D)\in\ZZ[D]$ and the lowest degree coefficient of $q(D)$ is a unit.

This condition defines an equivalence class with the relation
$$
\frac{p(D)}{q(D)} \sim \frac{p_1(D)}{q_1(D)} \mbox{ if and only if } p(D)q_1(D) = p_1(D) q(D).
$$

\noindent It follows directly from these definitions that $\ZZ[D]$ is a subring of $\ZZ(D)$,  which is itself a subring of $\ZZ((D))$.
\subsection{Convolutional codes over a finite ring}
A significant portion of the existing literature on convolutional codes over rings utilizes the framework of semi-infinite Laurent series to represent code sequences \cite{fa,z,ha,ma}. Following this literature, we adopt this approach in our definition:
\noindent A linear convolutional code $\C$ of length $n$ over $\ZZ$ is defined as a $\ZZ((D))$-submodule of $\ZZ((D))^n$. Such a code can be characterized by the existence of a  polynomial generator matrix $G(D)\in \ZZ[D]^{k\times n}$ where the code $\C$ is the image of $G(D)$ over $\ZZ((D))$ :  $$\C=Im_{\ZZ((D))}G(D)=\{\vv(D)\in\ZZ((D))^n/\exists \vu(D)\in\ZZ((D))^k\ :\vv(D)=\vu(D) G(D)\}.$$
The matrix $G(D)$ is called generator matrix.
If moreover $G(D)$ is a full row rank matrix over $\ZZ((D))$, then   is called an {\bf encoder} of $\C$ and $\C$ is a free module, called {\bf free convolutional code}.\\
A generator matrix $G(D)\in \ZZ[D]^{k\times n}$ is defined to be {\bf noncatastrophic} if any infinite weight input $\vu(D)\in\ZZ((D))^k$ cannot produces a finite weight output, i.e.,
$$\mbox{wt}[\vu(D)G(D)]<+\infty\Longleftrightarrow \mbox{wt}[\vu(D)]<+\infty,$$
where the Hamming weight wt$[\vu(D)]$ of a sequence    $\vu(D)=(\vu_1(D),\ldots,\vu_k(D))\in\ZZ((D))^k$ is defined as the sum of the weights of its components : $$ \mbox{wt}[\vu(D)]=\displaystyle\sum_{i=1}^k \mbox{wt}[\vu_i(D)]$$ where $ \mbox{wt}[\vu_i(D)]$ represents the number of nonzero coefficients of the series $\vu_i(D)$.\\
  In contrast to convolutional codes over fields, catastrophicity is an inherent property of codes defined over $\ZZ$. Specifically, there exist convolutional codes over $\ZZ$
 
for which every polynomial encoder is catastrophic. Such codes are termed  catastrophic codes. For instance, the code over $\Z_4$  generated by the matrix $\left(
                                                             \begin{array}{cc}
                                                               1+D & 3+D \\
                                                             \end{array}
                                                           \right)$
is a catastrophic code.\\  It is evident that  generator matrices not full row rank are  catastrophic. Therefore, convolutional codes that are not free are catastrophic codes.

\subsection{$p$-basis and Reduced $p$-basis}

\noindent  This subsection introduces the concepts of $p$-basis and reduced $p$-basis, which are fundamental for understanding the structure of convolutional codes over
$\ZZ$. These concepts were initially developed for numerical vectors over $Z_{p^r}$ \cite{v}  and extended to polynomial vectors  \cite{m}.\\
Every  element $a$ in $\ZZ$  has a unique $p$-adic expansion of the form $a=a_0+pa_1+\ldots+p^{r-1}a_{p-1}$  where the coefficients $a_i$ are in the set $\{0,\dots,p-1\}$. We will denote this set  throughout this work as $\A$. 

\noindent  Given a set of polynomial vectors $\vv_1(D), \dots, \vv_k(D) \in \Z^n_{p^r}[D]$,
\begin{description}
  \item[$\bullet$]  A {\bf $\boldsymbol p$-linear combinaison} of $\vv_1(D),\ldots,\vv_k(D)$ is a vector $\displaystyle \sum_{j=1}^k a_j(D) \vv_j(D)$,
where $a_j(D)$ are polynomial in $\A [D]$ ($\A[D]$ stands for the set of polynomials with coefficients in $\A$).
\item [$\bullet$] The {\bf $\boldsymbol  p$-span} of a set of polynomial vectors  ${\vv_1(D), \dots, \vv_k(D)}$, denoted by  $p$-span$[\vv_1(D), \dots, \vv_k(D)]$, is the set of all $p$-linear combinations of these vectors. For  a $k\times n$ matrix $G(D)$ with rows  $g_1(D),\ldots,g_k(D)$, we define  $p\mbox{-}span[G(D)]$ as  $p\mbox{-}span[g_1(D),\dots,g_k(D)]$.
\end{description}
An ordered set of vectors $(\vv_1(D), \dots, \vv_k(D))$ in $\Z_{p^r}[D]^n$ is called a {\bf $\boldsymbol{p}$-generator sequence} if, for each $i$ from 1 to $k-1$, $p\vv_i(D)$ 
can be expressed as  a $p$-linear combination of $\vv_{i+1}(D), \dots, \vv_k(D)$,
and $p\vv_k(D)=0$.
As  proven in
\cite{m} that
if $(\vv_1(D), \dots, \vv_k(D))$ is a $p$-generator sequence in $\Z_{p^r}[D]^n$ it holds that
$$
p\mbox{-}span[\vv_1(D), \dots, \vv_k(D)]=span[\vv_1(D), \dots, \vv_k(D)],
$$
and consequently $p\mbox{-}span[\vv_1(D), \dots, \vv_k(D)]$ is a $\ZZ[D]$-submodule of $\Z_{p^r}[D]^n$.\\
 The vectors $\vv_1(D), \dots, \vv_k(D)$  in $\Z_{p^r}[D]^n$ are defined as {\bf $\boldsymbol{p}$-linearly independent} if the only $p$-linear combination of
$\vv_1(D), \dots, \vv_k(D)$ that yields the zero vector  is the trivial one.
An ordered set of vectors $(\vv_1(D), \dots, \vv_k(D))$ which is a $p$-linearly independent $p$-generator sequence of a submodule $M$ of $\mathbb Z_{p^r}[D]^n$ is said to be a {\bf $\boldsymbol{p}$-basis} of $M$.\\
 The row degree of a nonzero polynomial vector $\vv(D)\in\ZZ[D]^n$ is defined as the highest degree of its nonzero components in $\ZZ[D]$. It is denoted by $d^\circ(\vv(D))$. The coefficient  vector in $\ZZ^n$ of the term $D^{d^\circ(\vv(D))}$ in $\vv(D)$ is called the {\bf leading row coefficient} vector of $\vv(D)$ and denoted by $\vv^{lc}$, likewise, the leading row coefficient matrix of a polynomial matrix $G(D)$ is denoted by $G^{lrc}$.
 \\ A $p$-basis $(\vv_1(D),\dots,\vv_k(D))$ of a module $M$ is called  a {\bf reduced $\boldsymbol {p}$-basis} if  the set of their leading row coefficient vectors  $(\vv_1^{lc},\dots,\vv_k^{lc})$ is $p$-linearly independent in $\ZZ^n$.\\  A key result,  proven in \cite{m}, establishes that every $\ZZ[D]$-submodule of $\ZZ[D]^n$ admits a reduced $p$-basis.\\

\subsection{Minimal $p$-encoder}
This section introduces the concept of a minimal p-encoder for convolutional codes over the ring $\mathbb{Z}_{p^r}$. To establish this notion, we first recall the following fundamental definitions from \cite{k}.

\begin{definition}\normalfont$(\boldsymbol{{p}\mbox{-}encoder}).$
 Let $\C$ be a convolutional code of length $n$ over $\ZZ$. A polynomial matrix $G(D) \in \mathbb{Z}_{p^r}[D]^{k \times n}$ whose rows form a $p$-linearly independent $p$-generator sequence is called a {\bf $p$-encoder} if
 $$\C=\{\vv(D)\in\ZZ((D))^n/\exists \vu(D)\in\A((D))^k:\vv(D)=\vu(D) G(D)\},$$
 where $\A((D))$ stands for the set of (semi-infinite) Laurent series with coefficients in $\A$.
\end{definition}
\noindent The following definition adapts the notion of a catastrophic encoder to the context of $p$-encoders.
\begin{definition}\normalfont$(\boldsymbol{Noncatastrophic}).$
A $p$-encoder matrix $G(D)\in\ZZ[D]^{k\times n}$ is said to be {\bf noncatastrophic} if any infinite weight input in $\A((D))^k$ cannot produce a finite weight output. Formally, this means :
$$\forall  \vu(D)\in\A((D))^k,\; \mbox{wt}[\vu(D)G(D)]<+\infty\Longleftrightarrow \mbox{wt}[\vu(D)]<+\infty.$$
\end{definition}

\begin{definition}\normalfont$(\boldsymbol{Delay\mbox{-}free\; and\; Reduced}).$
 Let $G(D)\in\ZZ[D]^{k\times n}$ be a $p$-encoder for a convolutional code  $\C$ of length $n$. The $p$-encoder $G(D)$ is said to be delay-free if the rows of $G(0)$ are $p$-linearly independent and is said to be reduced if its rows form a reduced $p$-basis for $\C$.  
 \end{definition}
 
 \begin{definition}\normalfont$(\boldsymbol{Minimal\mbox{-}p\;encoder}).$
  Let $\C$  be a noncatastrophic convolutional code of length $n$ over $\ZZ$ . Let $G(D)$ be a delay-free noncatastrophic $p$-encoder for $\C$, such that the rows of
$G^{lrc}$  are $p$-linearly independent in $\ZZ$. Then $G(D)$ is called a minimal $p-$encoder for $\C$.
\end{definition}
 
\noindent The existence of minimal $p$-encoders for codes over $\mathbb{Z}_{p^r}$ was raised as an open problem in the paper \cite{k}. The authors therein formulated the conjecture that every convolutional code over $\mathbb{Z}_{p^r}$ admits a minimal $p$-encoder.
\section{Polynomial matrices over a finite field }
\noindent This section presents foundational theorems and results concerning polynomial matrices over a finite field, which are essential for the subsequent development of this work. Notably, we define the polynomial $\Delta(G)$ associated with any polynomial matrix $G(D)$, a polynomial that plays a fundamental role in our analysis.
\\ Let $\F$ be a finite field. We denote by $\F[D],\F(D)$ and $\F((D))$  the sets of polynomial, rational  and Laurent series, respectively,  with coefficients in $\F$.\\ An $(n,k)$ convolutional code $\C$ over $\F$ is a $k$-dimensional $\F((D))$-subspace of $\F((D))^n$. A generator matrix of $\C$ is a $k\times n$ matrix $G(D)$ over $\F(D)$ whose rowspace is $\C$. If the entries of $G(D)$ are polynomials, then $G(D)$ is called a polynomial generator  matrix. \\ We say that two  generator matrices $G_1(D),G_2(D)\in\F[D]^{k\times n}$ are equivalent if they generate the same code. This is equivalent to the existence of a $k\times k$ invertible matrix $M(D)$ over  $\F(D)$ such that  $G_1(D)=M(D)G_2(D)$.\\
All matrices considered in this section are assumed to have full row rank over $\F(D)$.

\noindent Let $G(D)$ be a $k\times n$  polynomial matrix over $F[D]$.  We define the degree of a row $g_i(D)$ of $G$ as the maximum degree of its components. We define the degree and the  internal degree of $G(D)$ as follows
$$\mbox{intdeg G(D)}= \mbox{maximum degree of G(D)'s }k\times k \;\mbox {minors}$$
$$deg\; G(D)= \mbox{sum of the row degrees of G(D)}.$$
\begin{definition}\cite{rj}\rm

  A $k\times n $ polynomial matrix $G(D)$ is called {\bf basic} if, $G(D)$ has a right $\F(D)$ inverse, i.e, there is a $n\times k$ polynomial matrix $H(D)$ such that $$G(D)H(D)=I_k.$$
\end{definition}

\noindent The following  theorem characterizes basic  matrices.
\begin{theorem}\cite{rj}\label{t1}\rm

\item A $k\times n $ polynomial matrix $G(D)$ is basic if and only if the $\gcd$ of the $k\times k$ minors of $G(D)$ is 1.
\end{theorem}

\begin{definition}\rm
Let $G(D)$ be a $k\times n$ polynomial matrix. We define the polynomial  $\Delta(G)$ as the as the greatest common divisor $(\gcd)$ of the nonzero $k\times k$ minors of $G(D)$.
\end{definition}
\noindent If $G(D)\in \F[D]^{k\times n}$ has  full row rank, there is at least one nonzero minor of $G(D)$ . As a direct consequence of Theorem \ref{t1},  $\Delta(G)$ characterizes basic matrices : $G(D)$ is basic if and only if $\Delta(G)=1$.

\begin{proposition}\label{p2}\rm
Let $G(D)\in \F[D]^{k\times n}$  and $M(D)\in\F(D)^{k\times k}$ such that $M(D)G(D)\in\F[D]^{k\times n}$. Then $$\Delta(MG)=\lambda \Delta(G)\det M(D),$$
where $\lambda\in\F$  is such that $\lambda\det M(D)$ is a monic polynomial.
\end{proposition}
\begin{proof}\rm
Clearly the $k\times k$ minors of $M(D)G(D)$ are the $k\times k$ minors of $G(D)$ multiplied  by $\det M$ and the result holds immediately.
\end{proof}

\subsection{Catastrophic generator matrices}
\begin{definition}\normalfont A  $k\times n$ polynomial generator matrix $G(D)$ over $\F$ is said to be {\bf catastrophic} if there exists an infinite-weight vector $\vu(D)\in\F((D))^k$
such that the corresponding codeword $\vu(D)G(D)$ has finite weight.
\end{definition}
\noindent The following theorem, provided in \cite{jl},  characterizes   whether    a polynomial matrix $G(D)$  is noncatastrophic  based on its characteristic polynomial $\Delta(G)$.
\begin{theorem}\cite{jl}\label{j}\normalfont
A polynomial  matrix $G(D)$ is noncatastrophic if and only if  $\frac{1}{\Delta(G)}$ is of finite weight, i.e, $\Delta(G)$ is a power of $D$.
\end{theorem}
\noindent The following lemma is crucial for  understanding the sequel of this work, and its proof is based on a  proof of the Theorem \ref{j}  given by R.J. McEliece in \cite{rj}. 
\begin{lemma}\label{lx}\rm
For any matrix $G(D)\in\F[D]^{k\times n}$ there exists a matrix $M(D)\in\F(D)^{k\times k}$  where $\det M(D)$ is a non-zero constant in $\F$ and such that  $M(D)G(D)$ is a polynomial matrix which  admits a row whose entries are divisible by $\Delta(G)$. Moreover,  the matrix obtained from $M(D)G(D)$ by dividing this row by $\Delta(G)$ is basic.
\end{lemma}
\begin{proof}\rm
Let $\Delta(G)$ be factored into irreducibles polynomials over $\F[D]$ as $\Delta(G)=P_1(D)\dots P_s(D)$, where $P_i(D)$ are monic irreducible polynomials. Let  $\alpha_i$ be a zero of $P_i(D)$ in some extension field of $\F$.\\ Consider the matrix $G(\alpha_1)$. Each of the $k\times k$ minors of $G(\alpha_1)$ is zero, so rank $G(\alpha_1)<k$. Thus, there exists a $k\times k$ elementary matrix $M_1(\alpha_1)$ where the matrix $G_1(\alpha_1)=M_1(\alpha_1)G(\alpha_1)$ has one row  all zeros. We  assume that is the last row.  By replacing   $\alpha_1$  with the indeterminate $D$, we obtain $$G_1(D)=M_1(D)G(D),$$
where $M_1(D)$ is a unimodular $k\times k$ matrix and  the last  row of  $M_1(D)G(D)$ is divisible by $P_1(D)$. Let us  define the matrix $G'_1(D)$  obtained from $G_1(D)$ by dividing its  last  row by $P_1(D)$. This is equivalent to  multiplying  $G_1$ on the left by the diagonal matrix $$N_1=
                                                                                                                  \begin{pmatrix}
                                                                                                                    1 &  &  &  \\
                                                                                                                     & \ddots &  &  \\
                                                                                                                     &  & 1&  \\
                                                                                                                     &  &  & \frac{1}{P_1(D)} 
                                                                                                                  \end{pmatrix}
                                                                                                                .$$
We have $\Delta(G'_1)=\frac{\Delta(G)}{P_1(D)}=P_2(D)\dots P_s(D)$. We repeat the same procedure with $G'_1(D)$ and the next zero $\alpha_2$ of $\Delta(G)$.  Continuing this process, we obtain a sequence of unimodular matrices $M_1(D),\dots,M_s(D)$, for $i=1,\dots,s-1$ a sequence of diagonal matrices $N_i(D)$ with 1 in the $(k-1)$ first diagonal entries and $\frac{1}{P_i(D)}$ in the k-th diagonal. For $i=s$ we consider  $N_s(D)=
   \begin{pmatrix}
    1 &  &  &  \\
     & \ddots &  &  \\
     &  & 1&  \\
      &  &  & \frac{\Delta(G)}{P_s(D)} \\
     \end{pmatrix}
      .$ Let $M(D)=N_s(D)M_s(D)\dots N_1(D)M_1(D)$. Clearly the entries of the last row of $M(D)G(D)$ are divisible by $\Delta(G)$. Moreover, $\det M=cte\neq0$ which implies that $\Delta(G)=\Delta(MG)$.\\
 Finally, it is clear that the matrix obtained by dividing the last row of $M(D)G(D)$ by $\Delta(G)$ is basic.
\end{proof}

\noindent To illustrate the results of the last lemma, we present the  following examples.
\begin{examples}\normalfont  
\begin{enumerate}
\item Consider the matrix  $G(D)$ over  the field $\F_2$:
\begin{equation}\label{e}
G(D)=\left(
            \begin{array}{cccc}
              1+D & 0 & 1 & D \\
              D & 1+D+D^2 & D^2 & 1 \\
            \end{array}
          \right).
\end{equation}
By a simple calculation we have $\Delta(G)=1+D+D^2$. The polynomial  $1+D+D^2$  has a root $\alpha$ in $GF(2^2)$. Substituting $\alpha$ for $D$ in $G(D)$, we obtain $$G(\alpha)=\left(
            \begin{array}{cccc}
              1+\alpha & 0 & 1 & \alpha \\
              \alpha & 0 & \alpha^2 & 1 \\
            \end{array}
          \right).$$
          We note that $G(\alpha)$  has rank less than 2, and indeed if we multiply  $M(\alpha)=\left(
                                                                            \begin{array}{cc}
                                                                              1 & 0 \\
                                                                              1+\alpha & 1 \\
                                                                            \end{array}
                                                                          \right)$ by $G(\alpha)$ we obtain
\begin{equation}\label{e'}
G'(\alpha)=M(\alpha)G(\alpha)=\left(
                                \begin{array}{cccc}
                                  1+\alpha & 0 & 1 & \alpha \\
                                  0 & 0 & 0 & 0 \\
                                \end{array}
                              \right).
\end{equation}
Substituting $D$ for $\alpha$, in (\ref{e'}) we find

\begin{eqnarray*}
G'(D)&=&\left(
          \begin{array}{cc}
            1 & 0 \\
            1+D & 1 \\
          \end{array}
        \right)G(D)\\
&=&
\left(
         \begin{array}{cccc}
          1+D & 0 & 1 & D \\
           1+D+D^2 & 1+D+D^2 & 1+D+D^2 & 1+D+D^2  \\
            \end{array}
            \right).
\end{eqnarray*}
The last row of $G'(D)$ is already divisible by $\Delta(G)$ and the matrix obtained by dividing this row by $\Delta(G)$ is basic.

\item  The following example follows the steps seen in the proof of the last lemma. We consider the generator matrix over $\F_3$:
\begin{equation}\label{n1}
G(D)=\left(
\begin{array}{cccc}
2D & D & 2+D & 1 \\
1+D & 2+D & 1+D & 1 \\
1 & 0 & 1 & 2 \\
\end{array}
\right).
\end{equation}
If we denote the $3\times 3$ minor of $G(D)$ corresponding to columns $i,j$ and $k$ of $G(D)$ by $\Delta_{ijk}$, we have
                $$\Delta_{123}=2(1+D)(2+D),\;\;\Delta_{124}=2(1+D)(2+D),\;\;\Delta_{134}=2(1+D)(2+D),\;\;\Delta_{234}=0,$$
so that $\Delta(G)=(2+D)(1+D)$. Let us write $\Delta(G)=P_1(D)P_2(D)$,  $P_1(D)=2+D$ has the root $\alpha_1=1$ and $P_2(D)=1+D$ has the root $\alpha_2=2$. We begin by looking for a unimodular matrix $M_1(D)\in \F_3[D]^{3\times 3}$ such that one row, the third for example, of $M(D)G(D)$ will be divisible by $P_1(D)$. Substituting $D$ by 1 in equation (\ref{n1}), we obtain
$$G(1)=\left(
         \begin{array}{cccc}
           2 & 1 & 0 & 1 \\
           2 & 0 & 2 & 1 \\
           1 & 0 & 1 & 2 \\
         \end{array}
       \right).$$
 If we multiply the second row of $G(1)$ by $\alpha_1$ and add it to the third row, or alternatively multiply $G(1)$ by
$M_1(\alpha_1)=\left(
                 \begin{array}{ccc}
                   1 & 0 & 0 \\
                   0 & 1 & 0 \\
                   0 & 1 & 1 \\
                 \end{array}
               \right)$, we obtain
\begin{equation}\label{n2}
G_1(1)=M_1(1)G(1)=\left(
                                         \begin{array}{cccc}
                                           2 & 1 & 0 & 1 \\
                                           2 & 0 & 2 & 1 \\
                                           0 & 0 & 0 & 0 \\
                                         \end{array}
                                       \right).
\end{equation}
Substituting $D$ for $1$, in (\ref{n2}) and defining $G_1(D)=M_1(D)G(D)$, we obtain $$
G_1(D)=\left(
\begin{array}{ccc}
1 & 0 & 0 \\
0 & 1 & 0 \\
0 & D & 1 \\
\end{array}
\right)G(D)=\left(
\begin{array}{cccc}
2D & D & 2+D & 1 \\
1+D & 2+D & 1+D & 1 \\
(2+D)^2 & D(2+D) & (2+D)^2 & 2+D \\
\end{array}
\right).$$
Each entry in the last row of $G_1(D)$ is divisible by $P_1(D)$. The next step is to divide the third row of $G_1(D)$ by $P_1(D)$, or alternatively multiply 
$N_1(D)=\left(
          \begin{array}{ccc}
            1 & 0 & 0 \\
            0 & 1 & 0 \\
            0 & 0 & \frac{1}{2+D} \\
          \end{array}
        \right)$ by $G_1(D)$ and  obtain $$G_2(D)=N_1(D)G_1(D)=\left(
\begin{array}{cccc}
2D & D & 2+D & 1 \\
1+D & 2+D & 1+D & 1 \\
2+D & D & 2+D & 1 \\
\end{array}
\right).$$
As $M_1(D)$ is a unimodular matrix, we have $$\Delta(G_2)=\frac{\Delta(G)}{P_1(D)}=P_2(D)=1+D.$$ We repeat the same procedure with the matrix $G_2(D)$ and the polynomial $P_2(D)$. We obtain the unimodular matrix $M_2(D)=\left(
                                                    \begin{array}{ccc}
                                                      1 & 0 & 0 \\
                                                      0 & 1 & 0 \\
                                                      D & 0 & 1 \\
                                                    \end{array}
                                                  \right)$ where the third row of the matrix

$$M_2(D)G_2(D)=\left(
\begin{array}{cccc}
2D & D & 2+D & 1 \\
1+D & 2+D & 1+D & 1 \\
2(1+D)^2 & D(1+D) & (2+D)(1+D) & 1+D\\
\end{array}
\right)$$ is divisible by $P_2(D)$. Finally, as seen in the proof of the lemma, we take the matrix
$$N_2(D)=\left(
\begin{array}{ccc}
1 & 0 & 0 \\
0 & 1 & 0 \\
0 & 0 & \frac{\Delta(G)}{P_2(D)} \\
\end{array}
\right)=
\begin{pmatrix}
1 & 0 & 0\\
0 & 1 & 0\\
0 & 0 & 2+D
\end{pmatrix}
$$ and consider the matrix $$M(D)=N_2(D)M_2(D)N_1(D)M_1(D)=\left(
                                                                   \begin{array}{ccc}
                                                                     1 & 0 & 0 \\
                                                                     0 & 1 & 0 \\
                                                                     D(2+D) & D & 1 \\
                                                                   \end{array}
                                                                 \right).$$
Clearly $\det M(D)=1$ and we have
$$M(D)G(D)=
\begin{pmatrix}
2D & D & 2+D & 1 \\
1+D & 2+D & 1+D & 1 \\
2(1+D)^2(2+D) & D(1+D)(2+D) & (1+D)(2+D)^2 & (1+D)(2+D) \\
\end{pmatrix},$$
where the last row is divisible by $\Delta(MG)=\Delta(G)$ and the matrix obtained by dividing this row by $\Delta(G)$ is basic.
\end{enumerate}
\end{examples}

\begin{remark}\label{r11}\rm
Let $G(D)$ be a polynomial generator matrix over $\F$.
\begin{enumerate}
\item
 For any irreducible polynomial $P(D)$ that divides $\Delta(G)$,  there exists a unimodular matrix $M(D)$ such that the product $M(D)G(D)$ has one row  divisible  by $P(D)$.
\item Conversely, any polynomial devising a row of $G(D)$ is a divisor of $\Delta(G)$.
\item As a consequence, if $G(D)$  has a row  divisible by $\Delta(G)$,  then no  non-constant polynomial can divide  another row of $G(D)$.
\end{enumerate}
\end{remark}

\noindent The next theorem  characterizes  the infinite weight input sequences that produce a finite weight output when encoded by a catastrophic polynomial generator matrix. 
\begin{theorem}\label{px}\normalfont
Let  $G(D)$ be a $k\times n$ catastrophic  polynomial generator matrix,  and let $\vu(D)\in \F((D))^k$ be an infinite weight series. Then, the weight of the codeword $\vu(D)G(D)$ is finite if and only if  for any   invertible matrix $M(D)\in \F(D)^{k\times k}$, such that   $\det M(D)=c D^\alpha$, $c\in\F^*$  and $\alpha$ is a positive integer,    and  where the product $M(D)G(D)$ has  a j-th row (for some  $j\in\{1,\dots,k\}$) whose entries are divisible by $\Delta(G)$, the series  $\vu(D)M^{-1}(D)=(\vw_1(D),\dots,\vw_k(D))$  satisfies  the following  properties:

\begin{description}
  \item[i)] $\mbox{wt}[\vw_i(D)]<+\infty$, for $i=1,\dots,k$ where $i\neq j$.

  \item[ii)]  The component  $\vw_j(D)$ is of infinite weight and has the form $\frac{P(D)}{D^lQ(D)}$, where   $P(D),Q(D)\in\F[D]$ are coprime polynomials,  $Q(D)$ is a divisor of $\Delta(G)$,  and $l$ is a positive integer.
\end{description}

\end{theorem}
\begin{proof}\normalfont
   First, it is clear that an infinite weight series $\vu(D)$ which satisfies the assumption of the theorem produces a finite weight codeword.\\
   Let now $\vu(D)=(\vu_1(D),\ldots,\vu_k(D))$ be an infinite weight series in $\F((D))^k$ such that
  $\mbox{wt}[\vu(D)G(D)]<+\infty.$ Let $M(D)$ be an invertible matrix in  $\F(D)^{k\times k}$ satisfying the assumptions of the theorem. Without loss of generality, assume that the $k$-th (last) row of $M(D)G(D)$ is divisible by $\Delta(G)$. Then, the matrix $G'(D)$ obtained from $M(D)G(D)$ by dividing its $k$-row by $\Delta(G)$ is noncatastrophic. We have 
\begin{eqnarray*}
\vu(D)G(D)&=& \vu(D)M^{-1}(D)M(D)G(D)=(\vw_1(D),\dots,\vw_k(D))M(D)G(D)\\
&=&(\vw_1(D),\dots,\vw_{k-1}(D),\Delta(G)\vw_k(D))G'(D).
\end{eqnarray*}
 We have $\Delta(G')=\lambda \det M(D)$, where $\lambda$ is non-zero constant in $\F$, this means that the matrix $G'(D)$ is noncatastrophic. Since $\mbox{wt}[\vu(D)G(D)]<+\infty$ it follows that :  
 \begin{equation}\label{ah1}
 \mbox{wt}[(\vw_1(D),\dots,\vw_{k-1}(D),\Delta(G)\vw_k(D))]<+\infty.
 \end{equation}
This implies that
\begin{equation}\label{ah2}
 \mbox{wt}[\vw_i(D)]<+\infty \; \mbox{for}\; i=1,\ldots, k-1\;\; \mbox{and}\;\;    \mbox{wt}[\Delta(G)\vw_k(D)]<+\infty.
 \end{equation}
  Given that $\mbox{wt}[\vu(D)]=+\infty$ and $\det M^{-1}(D)$ is of finite weight,  then necessarily  $\vu(D)M^{-1}(D)=(\vw_1(D),\ldots,\vw_k(D))$ is of infinite weight. By combining (\ref{ah1}) and (\ref{ah2}), we have $\mbox{wt}[\vw_k(D)]=+\infty$ and $\mbox{wt}[\Delta(G)\vw_k(D)]<+\infty$.  This implies that $\vw_k(D)$ has  the following form
$\vu_k(D)=\frac{P(D)}{D^lQ(D)}$ where $P(D),Q(D)$ are coprime polynomials,  $Q(D)$ is a divisor of $\Delta(G)$ such that $\frac{1}{Q(D)}$ is of infinite weight   and $l$ is a positive integer. 

\end{proof}
\begin{example}\normalfont

 Let  $G(D)=
                                                  \begin{pmatrix}
                                                    1+D & D & 2 \\
                                                    2+2D & 2 & D 
                                                  \end{pmatrix}
                                                \in\Z_3[D]^{2\times3}$.\\
We have $\Delta(G)=2+D^2=(1+D)(2+D)$.
Let $\vu(D)=(\frac{D}{2+D},\frac{1}{2+D})$, we have $$\vu(D)G(D)=(\frac{D}{2+D},\frac{1}{2+D})
                                                  \begin{pmatrix}
                                                    1+D & D & 2 \\
                                                    2+2D & 2 & D \\
                                                  \end{pmatrix}
                                                =(1+D,\; 1+D,\; 0).$$ Thus, $\vu(D)$ is an infinite weight series which produces a finite weight output via $G(D)$.\\

\noindent Let the  matrices $M_1(D)=
                         \begin{pmatrix}
                           D & 1 \\
                           1 & 0 \\
                         \end{pmatrix}
                       $
 
 We have $\det M_1(D)=2$ and  $M_1(D)G(D)=\left(
                                    \begin{array}{ccc}
                                      2+D^2 & 2+D^2 & 0 \\
                                      1+D & D & 2 \\
                                    \end{array}
                                  \right)$
 Then  $M_1$  satisfies the assumptions of the theorem.
 The inverse of $M_1(D)$ is the matrix $M^{-1}(D)=
 \begin{pmatrix}
 0 & 1\\  & 2D
 \end{pmatrix}$,  and we have $$\vu(D)M^{-1}(D)=(\frac{1}{2+D},0).$$
The second component has finite weight while the first component has a dominator that is a devisor of $\Delta(G)$ fitting the form predicted by the theorem.\\
We can also take the matrix $M_2(D)=\left(
                                                             \begin{array}{cc}
                                                               1 & 0 \\
                                                               D & 1 \\
                                                             \end{array}
                                                           \right)$.
We have $\det M_2(D)=1$ and $M_2(D)G(D)=\left(
                                                                \begin{array}{ccc}
                                                                  1+D & D & 2 \\
                                                                  2+D^2 & 2+D^2 & 0 \\
                                                                \end{array}
                                                              \right)$.
  Then  $M_1$  satisfies the assumptions of the theorem. We have $M_2^{-1}(D)=
  \begin{pmatrix}
  1 & 0\\ 2D & 1
  \end{pmatrix}
  $. Then $$\vu(D)M_2^{-1}(D)=(0,\frac{1}{2+D}).$$
   The first component has finite weight while the second component has a dominator that is a devisor of $\Delta(G)$ as predicted by the theorem.                                                           
\end{example}

\section{Polynomial generator matrices   over $\ZZ$}
\noindent In this section we extend the definitions and key results concerning polynomial generator matrices  seen in the above section   to the finite ring $\ZZ$. A polynomial generator matrix, or encoder, is a $k\times n$ matrix over $\ZZ[D]$ whose rows are linearly independents over the field of rational functions $\ZZ(D)$.\\
\noindent We recall that two $k\times n$ polynomial  generator matrices  $G(D)$ and $G'(D)$ are equivalent if they generate the same code. As in the field case, this is equivalent to the existence of a $k\times k$ invertible matrix $T(D)$ over $\ZZ(D)$ such that $G(D)=T(D)G'(D)$ \cite{z}.\\
 We recall that a matrix $M(D)\in\ZZ[D]^{k\times k}$ is  unimodular  if $\det M$ is a unit in $\ZZ[D]$. This  is equivalent to $\det M=a+pQ(D)$, where $Q(D)\in\ZZ[D]$ and $ a$ is unit in $\ZZ$. Note that, $M(D)$ is unimodular if and only if its projection $\overline{M(D)}$ is unimodular over $\Z_p[D].$\\
 The next definition extends the definition of the polynomial $\Delta(G)$ to generator matrices over $\ZZ$.\\
 If $G(D)$ is a $k\times n $  polynomial generator matrix over $\ZZ$, then its projection $\overline{G(D)}$  is a polynomial generator matrix over $\Z_p$. This justifies the following definition.
\begin{definition}\normalfont
Let $G(D)$ be a  $k\times n$ polynomial generator  matrix over $\ZZ$. We define the polynomial $\Delta_p(G)$ as   the greatest common divisor $(\gcd)$ of the nonzero   $k\times k$  minors of the projected matrix $\overline{G(D)}$ in $\Z_p[D]$. That is  $$\Delta_p(G)=\Delta(\overline{G}).$$
\end{definition}
\noindent Our objective is to extend Lemma \ref{lx} to the ring case. For this, we need the following technical lemma, which is essential for the subsequent development. First we recall that a polynomial $P(D)\in\ZZ[D]$ is called regular if its projection $\overline{P(D)}$ into $\Z_p[D]$ is not the zero polynomial.
\begin{lemma}\label{l00}\normalfont
Let $P(D)$ and $Q(D)$ be two regular polynomials  in $\ZZ[D]$. Then, the polynomial $\overline{P(D)}$ divides the polynomial $\overline{Q(D)}$ in $\Z_p[D]$ if and only if, there exists an integer $i\in\{0,\ldots,r-1\}$ such that  for any integer $j\in\{i,\ldots,r-1\}$, $p^jP(D)$ divides $p^jQ(D)$.
\end{lemma}

\noindent Using the last lemma  and Lemma \ref{lx} we have the following.
\begin{theorem}\normalfont\label{lmx} Let $G(D)\in\ZZ[D]^{k\times n}$. Then there exists a matrix $M(D)\in\ZZ(D)^{k\times k}$, such that  $\det \overline{M(D)}$ is a nonzero constante in $\Z_p$, and where  $M(D)G(D)$  has  a row whose projection is divisible by $\Delta_p(G)$.\\
This means that there exists an integer $i_0\in\{0,\ldots,r-1\}$ where for any integer $i\in\{i_0,\ldots,r-1\}$ the matrix $P^iM(D)G(D)$ is polynomial and has a row divisible by a polynomial $Q(D)$ satisfying $\overline{Q(D)}=\Delta_p(G)$.\\
Moreover, by dividing this row by $Q(D)$, we  obtain a noncatastrophic matrix $H(D)$ such that for any $i\in\{i_0,\ldots,r-1\}$ we have  $$\mbox{span}[p^iH(D)]=\mbox{span}[p^iG(D)].$$
\end{theorem}

\begin{example}\normalfont
Let $G(D)=
\begin{pmatrix} 1+D & 9+D & 1+5D\\
D  & 5D^2 & 2+D^2
\end{pmatrix}\in\Z_{16}[D]^{2\times 3}$. We have $\Delta_p(G)=\Delta(\overline{G})=D(1+D)^2$. Let the matrix $M(D)=
\begin{pmatrix}
\frac{3}{1+D} & 0 \\
D & 3+D
\end{pmatrix}$. We have $\overline{M(D)G(D)}=
\begin{pmatrix}
1 & 1 & 1\\ 0 & D(1+D)^2& D(1+D)^2
\end{pmatrix}
$  where the second row   is divisible by $\Delta_p(G)$. It is easy to check  that the matrix $M(D)G(D)$ is not polynomial but the matrices  $2^2M(D)G(D)$ and 
$2^3M(D)G(D)$ are. Moreover, the second row of $ 2^3M(D)G(D)=\begin{pmatrix}
8 & 8 & 8\\
0 & 8D(1+D^2) & 8D(1+D^2)
\end{pmatrix}
$  is divisible by $\Delta_p(G)$. By dividing the second row by $\Delta_p(G)$ we obtain the matrix $\begin{pmatrix}
8 & 8 & 8\\
0 & 8 & 8
\end{pmatrix}=2^3\begin{pmatrix}
1 & 1 & 1\\
0 & 1 & 1
\end{pmatrix}.$ We take $H(D)=\begin{pmatrix}
1 & 1 & 1\\
0 & 1 & 1
\end{pmatrix}$ which is noncatastrophic  and we have $$\mbox{span}[2^3H(D)]=\mbox{span}[2^3G(D)].$$
\end{example}
\noindent Using Remark \ref{r11} and Lemma \ref{lmx} we have the following.

\begin{remark}\label{rmx}\normalfont
\begin{enumerate}
\item If the matrix $M(D)$ in the  Theorem \ref{lmx} happens to be unimodular, then $G(D)$ has an equivalent polynomial matrix $M(D)G(D)$ which has a row whose  projection over $\Z_p$ is divisible by $\Delta_p(G)$.
\item   An irreducible  polynomial  $P(D)\in\ZZ[D]$ divides a row of the matrix $p^iM(D)G(D)$  for some  unimodular matrix $M(D)\in \ZZ[D]^{k\times k}$, and some positive integer $i\in\{0,\dots,r-1\}$, if and only if, $\overline{P(D)}$ divides $\Delta_p(G)$.
    \item An irreducible polynomial $P(D)\in\ZZ[D]$ divides all minors of $G(D)$ if and only if,  there exists an unimodular matrix $M(D)$ such that $M(D)G(D)$ has a row divisible by $P(D)$.
\end{enumerate}

\end{remark}

\subsection{Catastrophic polynomial matrices}
A key characteristic of catastrophic polynomial generator matrices over $\ZZ$ lies in their projection over the residue field:
\begin{theorem}\cite{z}\label{j12}\normalfont\\
A polynomial matrix $G(D)$ over $\ZZ$ is catastrophic if and only if $\overline{G(D)}$ is catastrophic.
\end{theorem}
This immediately  implies that  a $k\times n$ polynomial matrix $G(D)$ is noncatastrophic if and only if $\Delta_p(G)$ is a power of $D$.\\

\noindent The following theorem is  essential  for the main result of this work. 
\begin{proposition}\label{cr}\normalfont
Let $G(D)$ be a $k\times n$ polynomial generator matrix over $\ZZ$. Then,  for all $i\in\{0,\ldots,r-1\}$, there exists a    matrix $G_i(D)\in\ZZ[D]^{k\times n}$ satisfying $$\mbox{span}[p^iG_i(D)]=\mbox{span}[p^iG(D)],$$  and for any series $\vu(D)\in \A((D))^k$ 
\begin{equation}\label{tt}
\mbox{wt}[p^i\vu(D)G_i(D)]<+\infty\; \Longleftrightarrow\; \mbox{wt}[\vu(D)]<+\infty.
\end{equation}
\end{proposition}
\begin{proof}\normalfont
  We beginning by $i=0$, we have $p^0G=G$. Suppose that $P_0(D)$ is an irreducible polynomial, not a power of $D$, that divides each of $k\times k$ minors of $G(D)$  (if not, we take $G_0=G$). Then there exists a nonsingular matrix $N_0(D)\in\ZZ(D)^{k\times k}$, $\det N_0(D)=cte$, such that $N(D)G(D)$ has a row divisible by $P_0$. We divides this row by $P_0$, we obtain a matrix $G_0(D)$ equivalent to $G(D)$ and where the minors have no irreducible polynomial as common divisor.\\
  Let now an infinite weight series $\vu(D)\in \A((D))^k$ such that $\mbox{wt}[\vu(D)G_0(D)]<+\infty$. Then 
  $$\mbox{wt}[\overline{\vu(D)}\overline{G_0(D)}]=\mbox{wt}[\vu(D)\overline{G_0(D)}]<+\infty.$$
Let $M(D)\in\Z_p(D)^{k\times k}$, $\det M(D)=cte$, such that the k-th row of $M(D)\overline{G_0(D)}$ is divisible by $\Delta(\overline{G_0})=\Delta_p(G_0)$. By Theorem  \ref{px}, the series $\vu(D)M^{-1}(D)$ has the following form  $$\vu(D)M^{-1}(D)=(w_1,\ldots,w_{k-1},\frac{P(D)}{D^lQ(D)})$$ where $\mbox{wt}[w_i]<+\infty$ and $Q(D)$ is a divisor of $\Delta_p(G_0)$ (not a power of $D$).\\
As $\mbox{wt}[\vu(D)G_0(D)]<+\infty$, we have $$\mbox{wt}[\vu(D)M^{-1}(D)M(D)G_0(D)]=\mbox{wt}[(w_1,\ldots,w_{k-1},\frac{P(D)}{D^lQ(D)})M(D)G_0(D)]<+\infty.$$ This implies that the k-th row of $M(D)G_0(D)$ is divisible by $Q(D)$ which is absurd. We conclude then: for any $\vu(D)\in\A((D))^k$, we have $$\mbox{wt}[\vu(D)G_0(D)]<+\infty\;\Longleftrightarrow\mbox{wt}[\vu(D)]<+\infty.$$ 
For $i=1$, Suppose that $P_1(D)$ is an irreducible polynomial, not a power of $D$, that divides each of $k\times k$ minors of $pG_0(D)$.  Then there exists a matrix $N_1(D)$, $\det N_1(D)$ is a nonzero constant in $\ZZ$, such that  $pN_1(D)G_0(D)$ has a row divisible by $P_1$. We divides this row by $P_1$ we obtain a matrix $pG_1(D)$ where the minors of $G_1(D)$ have no common divisor and  $$span[pG(D)]=span[pG_0(D)]=span [pG_1(D)].$$  Let now an infinite weight series $\vu(D)\in \A((D))^k$ such that $\mbox{wt}[\vu(D)G_1(D)]<+\infty$. We repeat the same reasoning seen for $i=0$, we obtain: for any $\vu(D)\in\A((D))^k$, $$\mbox{wt}[p\vu(D)G_1(D)]<+\infty\;\Longleftrightarrow\mbox{wt}[\vu(D)]<+\infty.$$
So, for any $i\in\{0,\ldots,r-1\}$ we find a matrix $G_i(D)$ satisfying \\ $\mbox{span}[p^iG_i(D)]=\mbox{span}[p^iG(D)]$ and for any $\vu(D)\in\A((D))^k$,  $$\mbox{wt}[p^i\vu(D)G_i(D)]<+\infty\; \Longleftrightarrow\; \mbox{wt}[\vu(D)]<+\infty.$$
\end{proof}
 The following example illustrate this result.
\begin{example}\normalfont
We consider  the ring $\Z_{27}$ and  the matrix $$G(D)=\left(
                                                  \begin{array}{cc}
                                                     2+7D^2 & 5+3D+19D^2+9D^3 \\
                                                  \end{array}
                                                \right).$$ By a simple calculus  we have   $$\Delta_3(G)=\Delta(\overline{G})=2+D^2=(1+D)(2+D).$$
The minors of $G(D)$ do not have an irreducible polynomial as common divisor, then we take $G_0(D)=G(D)$.

\begin{eqnarray*}
3G(D)&=&\left(
\begin{array}{cc}
 6+21D^2 & 15+9D+3D^2 \\
  \end{array}
 \right) \\
&=&3 \left(
 \begin{array}{cc}
  7(26+D)(1+D) & (26+D)(4+D) 
   \end{array}
   \right).
\end{eqnarray*}
The entries of  $3G(D)$ are divisible by $26+D$ where its projection  into $\Z_3[D]$ divides $\Delta_3(G)$. We take 
$G_1(D)=\begin{pmatrix}
7(1+D) & 4+D
\end{pmatrix}.$
We now look for the matrix $G_2$. We have 
$$9G_1(D)=9 
\begin{pmatrix}
7(1+D) & (4+D) 
\end{pmatrix}=9
\begin{pmatrix}
1+D & 1+D
\end{pmatrix}.$$
We take then $G_2(D)=\begin{pmatrix} 1 & 1
\end{pmatrix}.$\\
All matrices $G_0, G_1$ and $G_2$ satisfy the  properties given in the last proposition.
\end{example} 
\section{Application to convolutional codes }
In this section  we introduce a polynomial $\Delta_p(\C)$ for any free code $\C$ which is a generalisation of the polynomial $\Delta_p(G)$.
We will see too that this polynomial characterise  catastrophic free codes.
 
\begin{definition}\normalfont
A $k\times n$ polynomial generator matrix $G(D)$ is called {\it reduced internal degree matrix} (RIDM) if, among all polynomial matrices of the form $M(D)G(D)$, where $M(D)$ is a $k\times k$ non-singular matrix over $\Z_{p^r}(D)$, it has the minimum possible internal degree.
\end{definition}
Next theorem is a technical result that we needs for the sequel of this section.

\begin{theorem}{[Theorem XIII.6\cite{mc}}]\label{tech}\normalfont
 Let $P$ be a regular polynomial in $\Z_{p^r}[D]$. Then there is a monic polynomial $P_1$ with $\overline{P}=\overline{P_1}$ and a unit polynomial $P_2$ in $\Z_{p^r}[D]$ such that $P=P_1P_2.$

\end{theorem}

Next theorem is  fundamental  for the sequel of this section.
\begin{theorem}\label{imp}\normalfont
Let $G_1$ and $G_2$ be two equivalent reduced internal degree matrices, then $\Delta_p(G_1)=\Delta_p(G_2)$.
\end{theorem}
\begin{proof} Let $G_1(D)$ and $G_2(D)$ be two $k\times n$ matrices satisfying the conditions of the theorem. Then there is a $k\times k$ nonsingular matrix $M(D)$ over $\ZZ(D)$ such that $G_1(D)=M(D)G_2(D)$ which means  
$\Delta_p(G_1)=\det\overline{M}\Delta_p(G_2)$. Then to prove that $\Delta_p(G_1)=\Delta_p(G_2)$ it is sufficient to prove that $ \det\overline{M}$ is a nonzero constant in $\Z_p$. First we consider the particular case where $\det M$ is a  polynomial, which we denote by $Q(D)$. By Theorem \ref{tech},   we can write   $Q=P_1P_2$ where $P_1$ is a monic polynomial and $P_2$ is a unit a polynomial. Let $P$ be a monic irreducible divisor polynomial of $P_1$. We denote by $(\Delta_1, \ldots,\Delta_l,), l=C_n^k,$ the minors of $G_1(D)$. Then the minors of $G_2(D)$ are $(Q\Delta_1, \ldots,Q\Delta_l,)$. \\
As $P$ divides all minors of $G_2(D)$, by Remark \ref{rmx}, there exists an unimodular matrix $T(D)$ such that $T(D)G_2(D)$ has a row divisible by $P$.
By dividing this row by $P$ and as $P$ is a monic irreducible polynomial, we obtain a matrix $G'(D)$ such that   $$\mbox{intdeg}\; G'=\max (\deg Q \Delta_i)-\deg P=\mbox{intdeg}\; G_2-\deg P.$$ Since $G_2$ is RIDM, we have $\mbox{intdeg} G'\geq\mbox{intdeg} G_2.$  
Then we can conclude that  $\deg P=0$. This means  that any monic irreducible divisor of $P_1$ is  constant which implies that $P_1$ is a constant polynomial. Then  $\det M$ is a unit polynomial and then $\det\overline{ M}$ is a nonzero constant in $\Z_p$.\\
For the general case, let suppose that $\det M$ is not polynomial. Without loss of generality we can assume that  $\det M=R(D)=\frac{Q(D)}{P(D)}$ where $P(D)$ is a monic irreducible polynomial. Then the minors of $G_2(D)$ are $\frac{Q(D)}{P(D)}\Delta_i$. As $G_2(D)$ is a polynomial matrix, we conclude that $P$ divides all $\Delta_i$. This means that there exists an unimodular matrix $T(D)$ such that $T(D)G_1(D)$ has a row divisible by $P(D)$. By dividing this row by $P$, we obtain  a matrix $G'(D)$ equivalent to $G(D)$ satisfying $$\mbox{intdeg}\; G'=\max (\deg \Delta_i)-\deg P=\mbox{intdeg}\; G_1-\deg P.$$ Since $\mbox{intdeg}\; G'\geq  \mbox{intdeg}\; G$, we conclude that $P$ is a constant polynomial. Then $\det M$ is polynomial and the result holds from the first part of the proof.

\end{proof} 

\noindent  The following definition is justified by the last theorem.
\begin{definition}\normalfont
Let $\C$ be a free convolutional code over $\ZZ$. if $G(D)$  is   a reduced internal degree encoder for $\C$, then  We define $\Delta_p(\C)$ as   $$\Delta_p(\C)=\Delta_p(G).$$
\end{definition}
This definition is well-defined as $\Delta_p(G)$  is invariant under the choice of a reduced internal degree encoder, as shown in Theorem \ref{imp}.\\
The following theorem provides a characterization of noncatastrophic free codes, drawing upon Theorem \ref{j} and Theorem \ref{j12}.
\begin{theorem}\normalfont
A free code $\C$ over $\ZZ$ is noncatastrophic if and only if $\Delta_p(\C)$ is a power of $D$. Equivalently, this condition means that  $\frac{1}{\Delta_p(\C)}$ is of finite weight.
\end{theorem}

\subsection{Construction of minimal $p$-encoder}
The final result of this subsection is a constructive proof of the existence of a minimal $p$-encoder for any given convolutional code $\C$. To this end, we shall decompose $\C$ into simpler free submodules. Hence, we first address the case when $\C$ is a free module and then consider the general case.
\subsubsection{Case of free codes}
Let $\C$ be a free convolutional code over $\ZZ$ with $G(D)$ be an encoder for it. We define the following matrix
$$\widehat{G}(D)=
    \begin{pmatrix}
      G \\
      pG \\
      \vdots \\
      p^{r-1}G\\
    \end{pmatrix}.$$
It is straightforward to see that  $\widehat{G}(D)$ is a $p$-encoder for $\C$ and we have $$\C=span[G(D)]=p\mbox{-}span[\widehat{G}(D)].$$ 
\begin{lemma}\label{00}\normalfont
Let $G(D)$ be an encoder for a free code $\C$. Then, $\widehat{G}(D)$ is noncatastrophic, if and only if, $G(D)$ is noncatastrophic.
\end{lemma}
\noindent It is  possible now to state the main result of this section.
\begin{theorem}\normalfont\label{tt1}
Any free convolutional code over $\ZZ$ admits a minimal $p$-encoder.
\end{theorem}
\begin{proof}\normalfont
Let $\C$ be a free convolutional code  over $\ZZ$ and $G(D)\in\ZZ[D]^{k\times n}$ be a reduced internal degree   encoder of $\C$. Then we have $\Delta_p(G)=\Delta_p(\C)$. Let for $i=0,\ldots,r-1$, $G_i(D)$ the matrix deduced from $G(D)$ as describe in the Proposition \ref{cr}. As $G(D)$ is taken as  RIDM, then $G_0(D)=G(D)$.  Clearly the matrix  $G'(D)=\left(
                                                                         \begin{array}{c}
                                                                           G \\
                                                                           pG_1 \\
                                                                           \vdots \\
                                                                           p^{r-1}G_{r-1} \\
                                                                         \end{array}
                                                                       \right)$ is a $p$-encoder for $\C$. \\
Let $\vu(D)=(\vu_0(D),\dots,\vu_{r-1}(D))\in\A((D))^{kr}$ be an input sequence  which produces a finite weight output via $G'(D)$. This means that 
\begin{equation}\label{g}
\mbox{wt}[\vu(D)G'(D)]=\mbox{wt}[\sum_{i=0}^{r-1}\vu_i(D)p^iG_i(D)]<+\infty.
 \end{equation}
The equation \ref{g}, gives that $$\mbox{wt}[\overline{\vu(D)G'(D)}]= \mbox{wt}[\overline{\vu_0(D)}\overline{G_0(D)}]=\mbox{wt}[\vu_0(D)\overline{G(D)}]<+\infty.$$ 
Let suppose that $\mbox{wt}[\vu_0(D)]=+\infty$. Without loss of generality, we can suppose that the projection of $G(D)$ in $\Z_p$ admits a row divisible by $\Delta_p(G)$, (if not we replace $G$ by $MG$, where $M(D)\in\Z_p(D)^{k\times k}$, $\det M(D)=cte$ and   $M(D)G_0(D)$ has a row divisible by  $\Delta_p(G_0)$). Then, by Theorem \ref{px}, the first component of $\vu_0(D)$ is equal to $\frac{Q(D)}{D^mP(D)}$ where $P(D)$ is an irreducible divisor of $\Delta_p(G)$ and, $gcd(P,Q)=1$.\\ Let denote by $g(D)$ the first row of $G(D)$. The first component of $p\vu_1(D)G_1(D)$ is given by $pP_1(D)\vu_1^1(D)g(D)$, where $\vu_1^1$ is the first component of $\vu_1$ and $P_1$ is an irreducible  divisor of $\Delta_p(G_0)$. Similarly, the first component of $\vu_i(D)G_i(D)$ is $p^i\vu^1_i(D)P_1\ldots P_i$, where the polynomials $P_j$ are irreducibles divisors of $\Delta_p(G)$. Then, the first component of $\vu(D)G'(D)$ has the following form : $$\vu_1^1(D)G'(D)=\frac{Q(D)+pR(D)}{P(D)}g(D),$$ where $R(D)$ is a polynomial over $\ZZ$.
As $\mbox{wt}[\vu_1^1(D)G'(D)]<+\infty$ and $P(D)$ is not a divisor of $Q(D)+pR(D)$ we conclude that $P(D)$ divides the first row $g(D)$ of $G(D)$ and then divides all minors of $G(D)$ which is impossible. Then we have $\mbox{wt}[\vu_0(D)]<+\infty$.\\ It follows from equation \ref{g} that $$\mbox{wt}[p\vu_1(D)G_1(D)+\ldots+p^{r-1}\vu_{r-1}(D)G_{r-1}(D)]<+\infty,$$
which implies that $$\mbox{wt}[\vu_1(D)G_1(D)+\ldots+p^{r-2}\vu_{r-1}(D)G_{r-1}(D)]<+\infty.$$ We repeat the same reasoning as above, we conclude that $\mbox{wt}[\vu(D)]<+\infty.$  This proves that $G'(D)$ is noncatastrphic when the input sequence is in $\A((D))^{kr}$. Finally, we apply the row reduction algorithm of \cite{m} on the rows of $G'(D)$, we obtain a matrix $\widetilde{G}(D)$ whos rows form a reduced p-basis and then is a minimal p-encoder for $\C$.

\end{proof}
\begin{example}\normalfont
Let  $\C$ be the free convolutional code over  $\Z_{16}$  generated by the matrix $$G(D)=\left(
                                                                            \begin{array}{cccc}
                                                                             1+2D^2 & 1+D & 1+D & 1+D^2 \\
                                                                              D & 1+D & 15+3D & 2D^2 \\
                                                                            \end{array}
                                                                          \right).$$
  Clearly $G(D)$ is a reduced internal degree encoder for $\C$, then   $\Delta_2(\C)=\Delta_2(G)$.  By a direct calculus, We have $\Delta_2(G)=(1+D)^2$. Then, the code $\C$ is a catastrophic code. We will looking for the matrices $G_i$ defined above.  There is no irreducible polynomial as common divisor of the minors of $G(D)$, then we have $G_0=G$.  For the same reason,  we have $G_1(D)=G(D)$. For $i=2$,
 $$2^2G(D)=\left(
 \begin{array}{cccc}
 4+8D^2 & 4+4D & 4+4D & 4+4D^2 \\
 4D & 4+4D & 12+12D & 8D^2 \\
 \end{array}
 \right).$$ Let $M(D)=
 \left(
  \begin{array}{cc}
   1 & 1 \\
   1 & 0 \\
   \end{array}
    \right)$, we have
\begin{eqnarray*}
    2^2M(D)G(D)&=&\left(
 \begin{array}{cccc}
 4+4D+8D^2 & 8+8D & 0 & 4+12D^2 \\
 4D & 4+4D & 12+12D & 8D^2 \\
 \end{array}
 \right)\\
 &=&2^2\left(
 \begin{array}{cccc}
 (15+D)(3+2D) & 2(15+D) & 0 & 3(15+D)(1+D) \\
 D & 1+D & 3+3D & 2D^2 \\
 \end{array}
 \right)
 \end{eqnarray*}
where  the first  row  is divisible by $15+D$ . Then, we take  $G_2(D)$  the matrix obtained by dividing this row  by $15+D$. We have
$$2^2G_2(D)=2^2\left(
 \begin{array}{cccc}
 3+2D & 2 & 0 & 3(1+D) \\
 D & 1+D & 3(1+D) & 2D^2 \\
 \end{array}\right).$$

\noindent  Finally, for $i=3$,
$$2^3G_2(D)=\left(
 \begin{array}{cccc}
 8 & 0 & 0 & 8(1+D) \\
 8D & 8(1+D) & 8(1+D) & 0 \\
 \end{array}
 \right).$$
 We multiply $2^3G_2$ by the matrix $M(D)=\left(
                                          \begin{array}{cc}
                                            1 & 0 \\
                                            1 & 1 \\
                                          \end{array}
                                        \right)$, we have $$2^3M(D)G_2(D)=2^3\left(
 \begin{array}{cccc}
 1 & 0 & 0 & 1+D \\
 1+D & 1+D & 1+D & 1+D \\
 \end{array}
 \right).$$
 Then, by dividing the second row by $1+D$, we obtain the matrix $G_3(D)$ satisfying   $2^3G_3(D)=2^3\left(\begin{array}{cccc}
 1 & 0 & 0 & 1+D \\
 1 & 1 & 1 & 1 \\
 \end{array}
 \right)$ and  $span[2^3G_3(D)]=span[2^3G(D)]$.\\

\noindent We consider the following matrix 
$$G'(D)=\left(
                     \begin{array}{c}
                       G(D) \\
                       2G(D) \\
                       4G_2(D) \\
                       8G_3(D) \\
                     \end{array}
                   \right)=
\left(
  \begin{array}{cccc}
   1+2D^2 & 1+D & 1+D & 1+D^2 \\
    D & 1+D & 15+3D & 2D^2 \\
    2+4D^2 & 2+2D & 2+2D & 2+2D^2 \\
 2D & 2+2D & 14+6D & 4D^2 \\
 12+8D & 8 & 0 & 12(1+D) \\
 4D & 4(1+D) & 12(1+D) & 8D^2 \\
   8 & 0 & 0 & 8(1+D) \\
   8 & 8& 8 & 8 \\
  \end{array}
\right)$$
which is  a noncatastrophic $p$-encoder of $\C$.
The leading row coefficient matrix of $G'(D)$ is 
$$\begin{pmatrix}
2 & 0 & 0 & 1 \\
    0 & 0 & 0 & 2 \\
    4 & 0 & 0 & 2 \\
 0 & 0 & 0 & 4 \\
 8 & 0 & 0 & 12 \\
 0 & 0 & 0 & 8 \\
  0  & 0 & 0 & 8 \\
   8 & 8& 8 & 8 
 \end{pmatrix}
$$
By adding $D$ times the seven row to the sixth row, we obtain the matrix $ \widetilde{G}(D)$ given by $$\widetilde{G}(D)=
\begin{pmatrix}
1+2D^2 & 1+D & 1+D & 1+D^2 \\
    D & 1+D & 15+3D & 2D^2 \\
    2+4D^2 & 2+2D & 2+2D & 2+2D^2 \\
 2D & 2+2D & 14+6D & 4D^2 \\
 12+8D & 8 & 0 & 12(1+D) \\
 12D & 4(1+D) & 12(1+D) & 8D \\
   8 & 0 & 0 & 8(1+D) \\
   8 & 8& 8 & 8
 \end{pmatrix}$$
 whose rows are a reduced $p$-basis for $\C$. Finally, $ \widetilde{G}(D)$ is a minimal $p$-encoder for $\C$.
\end{example}

\subsection{Generalisation to non free codes}
\noindent The generalization of Theorem \ref{tt1} to arbitrary   code over $\ZZ$ is achieved  by decomposing such a code into simpler components.
 This decomposition was proved for block codes over $\ZZ$ in \cite{e} and extended  to convolutional codes in \cite{en}.
\begin{theorem}\label{t01}\cite{en}\normalfont
Let $\C$ be a convolutional code over $\ZZ$. Then, there exist free convolutional codes over $\ZZ$ $\C_0,\C_1,\dots,\C_{r-1}$ such that
$$\C=\C_0\oplus p\C_1\oplus\dots\oplus p^{r-1}\C_{r-1}$$ and $$\C_0+\dots+\C_{r-1}=\C_0\oplus\dots\oplus\C_{r-1}.$$
Note that the rank of the free codes $\C_i$ are a property of the code $\C$.
\end{theorem}
\noindent Let $\C=\C_0\oplus p\C_1\oplus\dots\oplus p^{r-1}\C_{r-1}$ and  $G_i(D)$ be an encoder for the free code $\C_i$, $i=0\dots r-1$. It's clear that the  matrix $$G(D)=\left(
                                                 \begin{array}{c}
                                                   G_0 \\
                                                   pG_1 \\
                                                   \vdots \\
                                                   p^{r-1}G_{r-1}\\
                                                 \end{array}
                                               \right)$$
is a generator matrix for $\C$. Then, the following matrix
$$\widehat{G}(D)=\left(
                     \begin{array}{c}
                       G_0 \\
                       p\left(
                          \begin{array}{c}
                            G_0 \\
                            G_1 \\
                          \end{array}
                        \right)
                        \\
                       \vdots\\
                       p^{r-1}\left(
                                \begin{array}{c}
                                  G_0 \\
                                  \vdots \\
                                  G_{r-1}\\
                                \end{array}
                              \right)
                        \\
                     \end{array}
                   \right)$$ is a $p$-encoder of $\C$.
Let $\C^i$ be the free code $\C_0\oplus\dots\oplus\C_i$. Then the matrix $$G^i=
                   \begin{pmatrix}
                     G_0 \\
                     \vdots \\
                     G_i \\
                   \end{pmatrix}$$
is an encoder for $\C^i$ and the matrix $\widehat{G}(D)$ has the following form :
$$
\widehat{G}(D)=
\begin{pmatrix}
G^0\\ pG^1\\ \vdots \\p^{r-1}G^{r-1}

\end{pmatrix}.
$$
By replacing each matrix $G^i$ by the corresponding matrix $H^i$ discribed in Proposition \ref{cr} we obtain the the following matrix 
$$\widetilde{G}(D)=\begin{pmatrix}
H^0\\ pH^1\\ \vdots \\ p^{r-1}H^{r-1}
\end{pmatrix}
.$$ Finally, we apply the row reduction algorithm on this matrix, the result is a minimal p-encoder for $\C$.
 Therefore, we are able to construct a minimal $p$-encoder for any code $\C$.
\begin{theorem}\normalfont
Every convolutional code over $\ZZ$ admits a minimal $p$-encoder.
\end{theorem}

\begin{example}\normalfont
Let's consider the ring $\Z_9$ and the code $\C=\C_0\oplus 3\C_1$ where $\C_0$ is the free code generated by $G_0=\left(
                                                                                          \begin{array}{cccc}
                                                                                            3D+3 & 5+D & 5+7D & 8+D \\
                                                                                            5+6D & 8+3D & 1+5D & 6+D \\
                                                                                          \end{array}
                                                                                        \right)$, and $\C_1$ is the free code generated by $G_1=\left(
                                                                                                                                 \begin{array}{cccc}
                                                                                                                                   1+4D & 4+7D & 7+D & 4+D \\
                                                                                                                                 \end{array}
                                                                                                                               \right).$
The matrix $$G(D)=\left(
                \begin{array}{c}
                  G_0 \\
                  3G_1 \\
                \end{array}
              \right)=\left(
                        \begin{array}{cccc}
                         3+3D & 5+D & 5+7D & 8+D \\
                        5+6D & 8+3D & 1+5D & 6+D \\
                          3+3D & 3+3D & 3+3D & 3+3D \\
                        \end{array}
                      \right)$$ is a generator matrix for $\C$.  A $p$-encoder for $\C$ is given by $$G'(D)=\left(
                                                                                                                    \begin{array}{c}
                                                                                                                      G_0 \\
                                                                                                                      3\left(
                                                                                                                         \begin{array}{c}
                                                                                                                           G_0 \\
                                                                                                                           G_1 \\
                                                                                                                         \end{array}
                                                                                                                       \right)
                                                                                                                       \\
                                                                                                                    \end{array}
                                                                                                                  \right).$$
There is no common divisor for the minors of $G_0$, then we keep this matrix.\\
Let $G^1(D)=\left(
            \begin{array}{c}
              G_0 \\
              G_1 \\
            \end{array}
          \right)=\left(
                    \begin{array}{cccc}
                      3+3D & 5+D & 5+7D & 8+D \\
                      5+6D & 8+3D & 1+5D & 6+D \\
                      1+4D & 4+7D & 7+D & 4+D \\
                    \end{array}
                  \right)$. We will look for a noncatastrophic matrix $H_1(D)$ such that $span[3H_1(D)]=span[3G^1(D)]$. We have $\Delta_3(G^1)=(1+D)^2(2+D)$,
then there exists a matrix $M(D)\in\Z_3(D)^{3\times 3}$, $\det M=1$, and where  $3M(D)G^1(D)$ has a row divisible by $\Delta_3(G^1)$. \\ Let $M(D)=\left(
                                              \begin{array}{ccc}
                                                \frac{1}{2+D} & 0 & 0 \\
                                                0 & \frac{1}{1+D} & \frac{1}{(1+D)^2} \\
                                                0 & 0 & (1+D)(2+D) \\
                                              \end{array}
                                            \right)$,
we have $$3M(D)G^1(D)=\left(
\begin{array}{cccc}
 0 & 3 & 3 & 3 \\
 0 & 0 & 6 & 3 \\
 3(1+D)^2(2+D) & 3(1+D)^2(2+D) & 3(1+D)^2(2+D) & 3(1+D)^2(2+D) \\
 \end{array}
 \right).$$ Let $3H_1(D)$ be the matrix obtained from $3M(D)G^1(D)$ by dividing its last row by $\Delta_3(G^1)$. The matrix $H_1(D)$ is then  noncatastrophic, and since $G_0(D)$ is  so also, the following  matrix $$\widetilde{G}(D)=\left(
                                      \begin{array}{c}
                                        G_0 \\
                                        3H_1 \\
                                      \end{array}
                                    \right)= \left(
                                               \begin{array}{cccc}
                                                 3D+3 & 5+D & 5+7D & 8+D \\
                                                 5+6D & 8+3D & 1+5D & 6+D \\
                                                 0 & 3 & 3 & 3 \\
                                                 0 & 0 & 6 & 3 \\
                                                 3 & 3 & 3 & 3 \\
                                               \end{array}
                                             \right)$$ is a noncatastrophic $p$-encoder for $\C$. Finally, the rows of the leading row coefficient matrix of $\widetilde{G}(D)$ are $p$-linearly independent, hence $\widetilde{G}(D)$ is a minimal $p$-encoder for $\C$.

\end{example}
\section{Conclusion}
\noindent \section{Conclusion}

The study of convolutional codes over finite rings presents many challenges compared to their counterparts over fields. A significant distinction is the existence of catastrophic codes over rings, where any polynomial  generator matrix can lead to an infinite weight input mapping to a finite weight output. This fundamental difference motivated the introduction of the novel concept of a "$p$-encoder" in \cite{k}, aiming to provide a noncatastrophic representation for any code over $\mathbb{Z}_{p^r}$. The authors in \cite{k} conjectured that every convolutional code over $\mathbb{Z}_{p^r}$ admits a minimal $p$-encoder, which implies that all such codes are noncatastrophic when the coefficients of the input series are in $\A$.

To address this conjecture, our work introduced a new polynomial, $\Delta_p(\mathcal{C})$, specifically tailored to characterize free codes over $\mathbb{Z}_{p^r}$. This allowed us to establish a necessary and sufficient condition for a free code to be catastrophic. Furthermore, we investigated the crucial relationship between infinite weight inputs (with coefficients in $\A$) producing finite weight outputs, demonstrating how the entries of such inputs are directly related to the polynomial $\Delta_p(\mathcal{C})$ and its divisors. These insights provided the foundation for our primary contribution: the systematic construction of minimal $p$-encoders, initially for free codes, and subsequently generalized to encompass all convolutional codes over $\mathbb{Z}_{p^r}$. This work thereby resolves the conjecture from \cite{k}, affirming the existence of a minimal $p$-encoder for every convolutional code over $\mathbb{Z}_{p^r}$.

\end{document}